\documentclass[10pt]{article}
\usepackage[margin=1in]{geometry}
\usepackage{amsfonts,amsmath,amssymb,amsthm}
\usepackage{comment,enumerate,hyperref}
\usepackage[ruled,vlined,linesnumbered]{algorithm2e}
\usepackage{algorithmic}
\newtheorem{theorem}{Theorem}
\newtheorem{corollary}[theorem]{Corollary}

\newtheorem{lemma}[theorem]{Lemma}

\theoremstyle{definition}

\title{Massively Parallel Approximate Distance Sketches}
\author{Michael Dinitz \\
Johns Hopkins University\\
mdinitz@jhu.edu
\and
Yasamin Nazari\\
Johns Hopkins University\\
ynazari@jhu.edu
}
\date{}
\begin{document}
\maketitle
\begin{abstract}
Data structures that allow efficient distance estimation (distance oracles, distance sketches, etc.) have been extensively studied, and are particularly well studied in centralized models and classical distributed models such as CONGEST.  We initiate their study in newer (and arguably more realistic) models of distributed computation: the Congested Clique model and the Massively Parallel Computation (MPC) model.  We provide efficient constructions in both of these models, but our core results are for MPC.  In MPC we give two main results: an algorithm that constructs stretch/space optimal distance sketches but takes a (small) polynomial number of rounds, and an algorithm that constructs distance sketches with worse stretch but that only takes polylogarithmic rounds.  

Along the way, we show that other useful combinatorial structures can also be computed in MPC.  In particular, one key component we use to construct distance sketches are an MPC construction of the hopsets of \cite{elkin2016}.  This result has additional applications such as the first polylogarithmic time algorithm for constant approximate single-source shortest paths for weighted graphs in the low memory MPC setting. 
\end{abstract}

\section{Introduction}

A common task when performing graph analytics is to compute distances between vertices.  This has motivated the study of shortest path algorithms in essentially every interesting model of computation.  We focus on two models which correspond to modern big-data graph analytics: Congested Clique~\cite{lotker2005} and Massively Parallel Computation (MPC)~\cite{beame2013}. The MPC model in particular has recently received significant attention, as it captures many modern data analytics frameworks such as MapReduce, Hadoop, and Spark. So since these are important models of distributed storage and computation, and computing distances in graphs is an important primitive, we have an obvious question: in MPC or Congested Clique, can we compute distances between nodes sufficiently quickly to support important graph analytics?

While one side effect of our techniques is indeed a state of the art algorithm for shortest paths in MPC, the focus of this paper is on getting around the limitations of these models by allowing preprocessing of the (distributed) graph.  We will first spend some time building a data structure known as \emph{approximate distance sketches} (or an approximate distance oracle), which will then let us (approximately) answer any distance query using only $0$, $1$, or $2$ rounds of network communication (depending on the precise model).  Thus after this preprocessing, anyone who is interested in analyzing the massive graph has access to approximate distances essentially for free, making this a powerful tool for distributed graph analytics.  Moreover, rather than inventing a brand new structure, we show that we can repurpose centralized data structures (in particular the Thorup-Zwick oracle~\cite{thorup2005}) by computing them efficiently in these new distributed models.  And since our algorithms are derived from centralized data structures we even allow for extremely efficient computation in addition to efficient communication.  

So our focus is on how to compute these data structures efficiently, since once they are computed distance estimates become fast and easy.  We show that in both the Congested Clique and the MPC models, we can compute oracles/sketches which essentially match the best centralized bounds in time that is only a small polynomial.  In MPC, we can go even further and compute slightly suboptimal sketches in time that is only polylogarithmic.  So while computing the data structure is still somewhat expensive, it is far more efficient than trivial approaches, and once it is computed, the analyst can receive approximate distances extremely quickly, allowing for low amortized cost or just the ability to do exploratory analysis without constantly waiting for expensive distance queries to complete.

\subparagraph*{Distance Oracles and Sketches.}
Even in many centralized applications, the time it takes to compute exact distances in graphs is undesriable, and similarly the memory that it would take to store all $n \choose 2$ distances is also undesirable.  This motivated Thorup and Zwick~\cite{thorup2005} to define the notion of an \emph{approximate distance oracle}: a small data structure which can quickly report an approximation of the true distance for any pair of vertices.  In other words, by spending some time up front to compute this data structure (known as the \emph{preprocessing} step) and then storing it (which can be done since the structure is small), any algorithm used in the future can quickly obtain provably accurate distance estimates.  

More formally, an approximate distance oracle is said to have \emph{stretch $t$} if, when queried on $u,v \in V$, it returns a value $d'(u,v)$ such that $d(u,v) \leq d'(u,v) \leq t \cdot d(u,v)$ for all $u,v \in V$, where $d(u,v)$ denotes the shortest-path distance between $u$ and $v$.   The important parameters of an approximate distance oracle are the size of the oracle, the stretch, the query time, and the preprocessing time.  For any constant $k$, Thorup and Zwick's construction (in the sequential setting) has expected size $O(k n^{1+1/k})$, stretch $(2k-1)$, query time $O(k)$, and preprocessing time $O(k m n^{1/k})$, where $n=|V|$ and $m=|E|$.  

Since~\cite{thorup2005}, there has been a large amount of followup work on improving the achievable tradeoffs, such as achieving query time of $O(1)$ with size $O(n^{1+1/k})$~\cite{wulff2013,chechik2014} or giving more refined bounds~\cite{patrascu2014,patrascu2012}.  However, with the notable exception of a very interesting construction due to Mendel and Naor~\cite{mendel2006}, the vast majority of followup work has essentially been refinements and improvements to the approach pioneered by Thorup and Zwick.  Thus understanding the Thorup-Zwick distance oracle is an important first step to understanding the limits and possibilities of distance oracles, and showing how to construct the Thorup-Zwick oracle in different computational models gives almost state-of-the-art bounds while also developing the basic tools and framework needed to design more sophisticated structures.

Importantly, the Thorup-Zwick distance oracle has the additional property that the data structure can be ``broken up" into $n$ pieces, each of size $O(k n^{1/k}\log n)$, so that the estimate $d'(u,v)$ can be computed just from the piece for $u$ and the piece for $v$ (the rest of the structure is unnecessary).  These are called \emph{distance sketches} or \emph{distance labelings}, and motivated Das Sarma et al.~\cite{sarma2015} to initiate the study of Thorup-Zwick distance sketches in distributed networks, and in particular in the CONGEST model of distributed computing~\cite{peleg2000}.
\paragraph{Models.} As mentioned, in modern graph analytics we usually abstract away the communication graph by assuming that the datacenter storing the graph is sufficiently well-provisioned.  This motivated two different but related models of distributed computation: Congested Clique \cite{peleg2000} and MPC \cite{beame2013}.  
In the Congested Clique model an input graph of $G = (V, E)$ is given, and initially each node $v \in V$ only knows its incident edges. However, the underlying communication graph is an undirected clique, and in each round every node can send a message of $O(\log n)$ bits to any other node.  This model was introduced by \cite{peleg2000}, and has been studied extensively in recent years. 
The second model that we consider is the \textit{Massively Parallel Computation}, or MPC model. This model was introduced by \cite{beame2013} to model MapReduce and other realistic distributed settings, and is more general than earlier abstractions of MapReduce proposed by \cite{karloff2010} and \cite{goodrich2011}. In this model there is an input of size $N$ which is arbitrarily distributed over $N/S$ machines, each of which has $S = N^{\epsilon}$ memory for some $0 < \epsilon < 1$.  In the standard MPC model, every machine can communicate with every other machine in the network, but each machine in each round can have total I/O of at most $S$. Specifically, for graph problems the total memory $N$ is $O(|E|)$ words. The low memory setting is the more challenging (but arguably more realistic) setting in which each machine has has $O(n^\gamma), \gamma <1$ memory, where $n= |V|$, which we denote by MPC($n^\gamma$). We also make the common assumption (e.g. \cite{roughgarden2018, beame2013}) that machines have unique IDs that other machines can use for direct communication.

\subsection{Our Results} 
In this paper we initiate the study of distance oracles and sketches in two popular computational models for ``big data'': Congested Clique and MPC.  In addition, we show that our techniques can be used to give the first sublinear  algorithm (and in fact polylogarithmic) for approximate single-source shortest paths for weighted graphs in (low memory) MPC, and moreover can be applied in straightforward ways to non-distributed models such as the streaming setting.  
We discuss our results for each model in turn.  At a high level, Congested Clique turns out to be relatively easy: we can essentially just combine the known CONGEST algorithm~\cite{sarma2015} with a slightly modified hopset construction.  For MPC, the natural approach is to simulate the Congested Clique algorithm, since it is known~\cite{behnezhad2018} that under certain density and memory conditions, Congested Clique algorithms can be simulated in MPC.  However, this simulation requires at least $\Omega(n)$ memory per machine.  Our task becomes much more challenging if we allow $o(n)$ memory per machine, which we refer to as the \emph{low memory} setting.  Designing algorithms for this setting forms the bulk of this paper. \paragraph{Congested Clique.} Since there is no memory restriction for Congested Clique, we assume that some node in the network is the \emph{coordinator} at which the entire distance oracle will be stored (i.e., the machine with which users will interact with the distributed system).  So at query time, the user can just query the coordinator locally (avoiding all network delay) rather than initiating an expensive distributed computation. The precise statements of our results are given in Appendix~\ref{app:congest_clique} and are somewhat technical, so for simplicity we state one particularly interesting corollary obtained by some specific parameter settings:  
\begin{theorem}
Given a weighted graph $G = (V, E, w)$, for all $k \geq 2$ and constant $\epsilon > 0$, we can construct a distance oracle with stretch $(1+\epsilon)(2k-1)$, (local) query time $O(k)$, and space $O(kn^{1+1/k} \log n)$ w.h.p.~in the Congested Clique model.  If $k = O(1)$, then the number of rounds for preprocessing is\footnote{The notation $\tilde{O}(f(n))$ stands for $O(f(n) \cdot \text{polylog}(f(n))$, e.g. it is suppressing polyloglog($n$) terms in $2^{\tilde{O}(\log n)}$.} $\tilde O(n^{1/k})$, and if $k = \Omega(\log n)$ then the number of rounds is $\tilde{O}(\log(n))$.      
\end{theorem} 

Note that after a limited amount of preprocessing, distance queries can be computed without any network access whatsoever.  Moreover, the computational query time is also extremely small, so these queries are extraordinarily efficient in the context of distributed algorithms. As an interesting extension, we show that the message complexity of computing this distance oracle can be reduced by adding an additional preprocessing step of computing a graph spanner.
\paragraph{MPC.} In Section \ref{sec:MPC} we discuss the MPC model, which is the heart of this paper. Since in the MPC model servers have small memory, it is impossible to fit an entire distance oracle at a single server as we did in the Congested Clique.  So we instead focus on distance sketches. After the preprocessing algorithm, for each node $v \in V$, a distance sketch of size $O(kn^{1/k}\log n)$ will be stored and mapped to a machine with key $v$ (this assumes that the memory at each server is at least $\Omega(k n^{1/k} \log n)$, which is reasonable in most settings). This means that after the preprocessing to construct these sketches, only two rounds of communication are needed for for approximating distance queries between a pair of nodes $u$ and $v$: one for sending requests for the sketches of $u$ and $v$ and one for receiving them. We give the following result:

\begin{theorem} \label{intro_thm:MPC_sketches}
Given a weighted graph $G=(V, E, w)$ with polynomial weights\footnote{This assumption can be relaxed using reduction techniques (e.g.~from \cite{elkin2016}) in exchange for extra polylogarithmic factors in the hopbound and construction time.} and parameters $\rho \leq \gamma \leq 1, 1/k \leq \rho, 0 < \epsilon<1$, we can construct Thorup-Zwick distance sketches with stretch ${(2k-1)(1+\epsilon)}$ and size $O(kn^{1/k}\log n)$ w.h.p.~in $\tilde{O}( \frac{1}{\gamma} \cdot n^{1/k} \cdot \beta)$ rounds of MPC$(n^\gamma)$, where ${\beta= \min(O(\frac{\log n}{\epsilon})^{\log(k)+k}, 2^{\tilde{O}(\sqrt{\log n})}})$. In particular, if $k=O(1)$ and $\epsilon$ is a constant, then w.h.p.~we require $\tilde{O}(n^{1/k})$ rounds, and if $k=\Theta(\log n)$ then w.h.p.~we require $2^{\tilde{O}(\sqrt{\log n})}$ rounds.
\end{theorem}

In the above theorem the distance sketches have the same guarantees as the centralized Thorup-Zwick distance oracles. However, in MPC a polynomial round complexity, while possibly of theoretical interest, is generally considered not practical.  So we give a different (but related) algorithm which achieves polylogarithmic round complexity, at the price of larger stretch.

\begin{theorem}\label{thm:polylog_distancesketches}
  Consider a graph $G=(V,E)$ where $m=\Omega(kn^{1+1/k} \log n)$, for any $k \geq 2$. Then there is an algorithm in MPC($n^{\gamma}$) (with $0 < \gamma < 1$) that constructs Thorup-Zwick distance sketches with stretch $O(k^2)$ and size $O(kn^{1/k}\log n)$ and with high probability completes in ${O(\frac{k}{\gamma} \cdot (\frac{ \log n \cdot \log k}{\epsilon} )^{\log k +k-1})} $ rounds.
\end{theorem}

As a side effect of our techniques (which we discuss more in Section~\ref{sec:techniques}), we immediately get an algorithm for computing approximate single-source shortest paths (SSSP) in the MPC model, which is the problem of finding the (approximate) distances from a source node to all other nodes.  Unlike Congested Clique, there do not seem to be any known nontrivial results for this problem in MPC. We first give an algorithm which computes a $(1+\epsilon)$-approximation in $n^{o(1)}$ time. Then we show that  we can compute an $O(1)$-approximation in only polylogarithmic time, if we make an additional assumption about the density of the input graph. We will prove the following theorem in Section \ref{sec:SSSP}:

\begin{theorem} \label{intro_thm:MPC_SSSP}
Given a weighted undirected graph $G=(V, E, w)$ with polynomial weights, a source node $s \in V$, and $0 < \gamma \leq 1, 0 < \epsilon<1$ we can compute $(1+\epsilon)$-approximate SSSP w.h.p.~in ${O(\frac{1}{\gamma})\cdot 2^{\tilde{O}(\sqrt{\log n})}}$ rounds of MPC with $\Theta(n^\gamma)$ memory per machine.
 Moreover, if $|E| \geq \Omega(n^{1+1/k}\log(n))$, we can compute $4k(1+\epsilon)$-approximate SSSP~in ${O(\frac{1}{\gamma} \cdot (\frac{ \log n \cdot \log k}{\epsilon} )^{\log k+k-1})} $ rounds of MPC($n^\gamma$), where $1/k < \gamma \leq 1, k \geq 2$. In particular, for $k=O(1)$ the algorithm runs in ${O(\frac{1}{\gamma} \cdot (\frac{ \log n}{\epsilon} )^{O(1)})}$ rounds.
\end{theorem}

Note that while the round complexity is polylogarithmic, it may still be somewhat slow for certain applications: an analyst who has to wait polylogarithmic rounds for every distance query would essentially be unable to perform any analysis which depended on large numbers of distance queries.  On the other hand, our main results on distance sketches allows us to pay this round complexity only once, for constructing the sketch.
\paragraph{Streaming.}Finally, we provide an algorithm for constructing distance oracles in the multi-pass streaming model. This is essentially a side-effect of our main results for Congested Clique and MPC, but we include it for completeness.  Our general results can be found in Appendix \ref{app:streaming}. For the specific settings of constant or logarithmic stretch, we have:
\begin{corollary}
Given a graph $G=(V, E,w)$, there exists a streaming algorithm that constructs a Thorup-Zwick distance oracle of stretch $(2k-1)(1+ \epsilon)$ of size $O(kn^{1+1/k} \log n)$  w.h.p. and expected space $O(n^{1+1/k} \cdot \log^2 n)$, such that if $k=O(1)$, w.h.p.~we require $O(\log^k n)$ passes , and if $k=\Omega(\log n)$, w.h.p.~we require $2^{\tilde{O}(\sqrt{\log n})}$ passes.
\end{corollary}
Note that in case of $k=\Omega(\log n)$ we are in the so-called \textit{semi-streaming} setting in which the total memory used is $O(n \cdot \text{polylog } n)$. 
\subsection{Our Techniques} \label{sec:techniques}
Our main approach is to combine constructions of \emph{hopsets} with efficient distributed constructions of Thorup-Zwick distance oracles/sketches.  In particular, Das Sarma et al.~\cite{sarma2015} showed that Thorup-Zwick sketches could be computed in the CONGEST model, but the time depended on the graph diameter.  So all that we really need to do is to reduce the diameter of the graph, since any CONGEST algorithm also works in the Congested Clique.  This is what hopsets do: we discuss them in more detail in Section~\ref{sec:algs}, but informally they allow us to reduce the diameter of the graph while preserving distances by adding in a carefully chosen set of weighted ``shortcut" edges. Hopset constructions for the Congested Clique were given by Elkin and Neiman~\cite{elkin2016} (and more recently by\cite{censor2019}) so for Congested Clique we can essentially just combine result of \cite{elkin2016} (or \cite{censor2019}) with \cite{sarma2015} to get our result (modulo a small number of technicalities).  

Moving to MPC introduces some significant technical difficulties, particularly when the space per machine is $o(n)$. Neither \cite{sarma2015} nor \cite{elkin2016} are written with MPC in mind, so we cannot simply ``black-box" them as we could (mostly) in the Congested Clique. However, not surprisingly, both \cite{sarma2015} and \cite{elkin2016} use as a fundamental primitive a ``restricted" version of the classical Bellman-Ford shortest-path algorithm that ends early, and it turns out that implementing this restricted Bellman-Ford is the main (although not the only) technical hurdle in adapting both of them to the MPC model. 

When implementing restricted Bellman-Ford in low-memory MPC, the main difficulty is that since the memory at each server is $o(n)$, a single server cannot ``simulate" a node in Bellman-Ford. It takes many machines to store the edges incident on any particular node, so we need to show that it is possible for many machines to simulate a single node in MPC without too much overhead.   We show that this is indeed possible: Bellman-Ford and related algorithms can be implemented in low-memory MPC with very little additional overhead.  Once we develop this tool, we argue that the hopsets of \cite{elkin2016} can be constructed in low-memory MPC with essentially the same complexity as in the Congested Clique. Our implementation of Bellman-Ford and this hopset construction, as well as a few other primitives we develop for low-memory MPC (e.g., finding minimum or broadcasting on a range of machines), may be of independent interest.

Even after using hopsets, we would still need polynomial time for constructing constant stretch distance sketches. We overcome this issue and improve the running time using two ideas. First, we show that by relaxing the model to allow small additional total memory (either through extra space per machine or additional machines), we can run our algorithms in polylogarithmic number of rounds. So we just need to argue that there is a way of obtaining extra memory without actually changing the model assumptions.  This is our second idea: by constructing a spanner we can sparsify the graph while keeping the memory per machine and number of machines the same. Thus from the perspective of the spanner, it will appear that we do indeed have ``extra" memory.  The idea of sparsifying the input to obtain extra resources has already proved to be powerful in related contexts (for example, \cite{friedrichs2018} recently used spanners to give a work-efficient PRAM metric embedding algorithm).  To the best of our knowledge, though, this idea has not yet appeared in the MPC graph algorithms literature.  

\subsection{Related Work} \label{sec:related}
Distributed constructions of distance oracles and sketches have been studied extensively in the CONGEST model~\cite{sarma2015,lenzen2013,elkin2016podc}.  All of these algorithms have running times dependent on the graph diameter, while our algorithms run in time independent of the graph diameter.  To the best of our knowledge, constructing distance oracles/sketches has not previously been studied for the Congested Clique or the MPC model. Similarly, hopsets have been used extensively in various models of computation for solving approximate SSSP (\cite{henzinger2016, elkin2016}). Our result on hopset construction in low memory MPC also gives the first (approximate) SSSP algorithm in this model for weighted graphs (in Congested Clique there are more results known~\cite{elkin2016, henzinger2016, becker2017, censor2019}, but these do not translate obviously to MPC when there is sublinear memory per machine).
In a recent result, \cite{censor2019} gave an efficient Congested Clique algorithm that constructs hopsets of size $\tilde{O}(n^{3/2})$ with hopbound $O(\log^2(n)/\epsilon)$. Their hopsets are a special case of hopsets of \cite{elkin2016}. In Appendix \ref{app:congest_clique} we explain how their algorithm applies to our Congested Clique result.

In the PRAM model, shortest path computation is well studied (e.g. \cite{cohen2000, elkin2016}), and it is known that many PRAM algorithms can be simulated in the MPC model (\cite{karloff2010, goodrich2011}). However, most of these algorithms use $\omega(|E|)$ number of processors, in which case the simulations of \cite{karloff2010} and \cite{goodrich2011} do not directly apply as they assume that the number of processors is at most the input size. As we argue in Section \ref{sec:extra_memory} we will still utilize an extension of this simulation. Another recent result for APSP in MapReduce by \cite{hajiaghayi2019} also has the same drawback of using $\omega(n^2)$ processors. Result of \cite{hajiaghayi2019} is based on matrix multiplication techniques, which are also well-studied in the PRAM model for computing APSP.

Finally, we note that distance problems have also been studied in related models such as the $k$-machine model (\cite{klauck2015}). In this model \cite{klauck2015} shows a low bound of $\Omega(n/k)$ for computing shortest paths, where $k$ is the number of machines. To the best of our knowledge, the exact connection between this model and the MPC model has not yet been studied\footnote{In the $k$-machine model, generally the number of machines considered is small. The computational power of this model therefore seems very different from the low-memory MPC setting, where there are many machines (more than $n$), but each one has small memory. Moreover, the $k$-machine model does not bound the space on each machine and the IO bound is slightly different from MPC.}.

\section{Preliminaries and Notation}

\subsection{Notation}
In a given weighted graph $G=(V,E)$, we denote the (weighted) distance between a pair of nodes $u,v \in V$ by $d_G(u,v)$. We may drop the subscript $G$ when there is no ambiguity. We define the \textit{$h$ hop-restricted} distance between $u$ and $v$ to be the weight of the shortest path between $u$ and $v$ that uses at most $h$ hops and denote this by $d^h(u,v)$. 

We will denote the set of neighbors of a node $v \in V$ by $N(v)$. In a weighted graph $G$, we define the \emph{shortest-path diameter} of $G$, denoted by $\Lambda$, to be the maximum over all $u,v \in V$ of the number of edges in the shortest $u-v$ path (so if the graph is unweighted this is the same as the diameter, but in weighted settings it can be larger than the unweighted diameter).
 Finally, a $t$-spanner of $G$ is simply a subgraph which preserves distances up to a multiplicative $t$ factor.

\subsection{Algorithmic Building Blocks} \label{sec:algs}
In this section we describe the algorithms of \cite{thorup2005}, \cite{sarma2015} and \cite{elkin2016}, that we will use in next section.
\paragraph{Thorup-Zwick Distance Oracle.}\label{sec:TZ}
In this section, we briefly describe the centralized construction of the well-known Thorup-Zwick distance oracle~\cite{thorup2005}. Given an undirected weighted graph $G=(V,E,w)$ and $k > 1$, in the preprocessing phase of their algorithm they first create a hierarchy of subsets $A_0,A_1,...,A_{k_1}$ by sampling from nodes of $V$ in the following manner: set $A_0=V$, and for $1 \leq i \leq k-1$, add every node $v \in A_{i-1}$ to the set $A_i$ independently with probability $n^{-1/k}$. Set $A_k = \emptyset$ and for all $u \in V$ define $d(u,A_k)=\infty$.
Let $B_i(u) =\{ w \in A_i: d(u,w) < d(u, A_{i+1}) \}$ for all $u \in V$ and $0 \leq i \leq k-1$, where $d(u,A_i)$ is the minimum distance between $u$ and a node in the set $A_i$, and set $B(u)= \cup_{i=0}^{k-1} B_i(u)$.  We also denote the node that has the minimum distance to $u$ among all nodes in $A_i$ by $p_i(u)$ and call this the $i$-center of $u$, and so $d(u,A_i)=d(u,p_i(u))$.  The distance sketch for $u$ consists of $\{p_i(u)\}_{i=0}^k$, the set $B(u)$, and the corresponding distances between these nodes and $u$.  The distance oracle is just the union of the sketches for all $u \in V$.  Thorup and Zwick showed that this data structure has size $O(kn^{1+1/k} \log n)$ w.h.p., and access to these sketches is enough for approximating distances between every pair of vertices in $O(k)$ time with stretch $2k-1$. In all the settings we consider, after preprocessing the distance oracle/sketches, we can \textit{locally} perform the query algorithm of \cite{thorup2005} in $O(k)$ time. For completeness, we briefly review the query algorithm in Appendix \ref{app:query}. 

Next, we explain a distributed construction of Thorup-Zwick distance \textit{sketches} as described by Das Sarma et al.~\cite{sarma2015} for the CONGEST model. The sampling phase can easily be done in distributed settings. Then for finding $p_i(v), 1 \leq i \leq k$ for all nodes $v \in V$, we will do the following: in iteration $i$, define a virtual source node $s_i$, and for all nodes in $u \in A_i$ add an edge between $u$ and $s_i$ where $w(u,s_i)=0$. Then we will only need to run the Bellman-Ford algorithm from $s_i$, and after $O(k \Lambda)$ time every node $u \in V$ knows $p_i(u)$ and $d(u,A_i)$. Finally, for all $1 \leq i \leq k$ we need to compute the distance from $w \in A_i \setminus A_{i+1}$ to all the nodes $v$ for which $w \in B(v)$. Simply running a distributed Bellman-Ford independently from all the sources $w \in A_i \setminus A_{i+1}$ would be slow since due to congestion limit on each edge we cannot run all these in parallel at the same time. However, \cite{sarma2015} argue that this can be done in $O(\Lambda \cdot kn^{1/k} \log n)$ rounds in total (w.h.p), since each node $v$ needs to forward messages in the runs of Bellman-Ford algorithm for a source $w$ only if $w \in B(v)$. This means that, roughly speaking, each node $v$ participates in $|B(v)|= O(kn^{1/k} \log n)$ runs of Bellman-Ford. Then by a simple round-robin scheduling scheme they show that running these Bellman-Fords for all sources in ${A_i\setminus A_{i+1}}$ can be done in $O(\Lambda \cdot kn^{1/k} \log n)$ without violating the congestion bound on each edge. For completeness we include a more detailed version of this algorithm in Algorithm \ref{alg:distanceoracle} in Appendix \ref{app:congest_clique}.
\paragraph{Hopsets.} For parameter $\epsilon , \beta >0$, a graph $G_H = (V, H, w_H)$ is called a $(\beta,\epsilon)$-hopset for the graph $G$, if in graph $G'=(V, E\cup H,w')$ obtained by adding edges of $G_H$, we have $d_G(u, v) \leq d^{\beta}_{G'} (u, v) \leq (1+\epsilon) d_G (u, v)$  for every pair $u, v \in V$ of vertices. The parameter $\beta$ is called the \textit{hopbound} of the hopset.


We first give a high level overview of the (sequential) hopset construction of \cite{elkin2016} here. In their algorithm, they consider each distance scale $(2^k, 2^{k+1}], k=0,1,2,...$ separately. For a fixed distance scale $(2^k,2^{k+1}]$ the algorithm consists of a set of \textit{superclustering}, and \textit{interconnection} phases. Initially, the set of clusters is $\mathcal{P}=\{ \{v\}_{v \in V}\}$. Each cluster in $C \in \mathcal{P}$ has a cluster center which we denote by $r_C$. The algorithm uses a sequence $\delta_1,\delta_2,...$ of distance thresholds and a sequence $\deg_1,\deg_2,...$ of degree thresholds that determines the sampling probability of clusters. At the $i$-th iteration, every cluster $C \in \mathcal{P}$ is sampled with probability $1/\deg_i$. Let $S_i$ denote the set of sampled clusters.
 Now a single shortest-path exploration of depth $\delta_i$ (weighted) from the set of centers of sampled clusters $R=\{r_C \mid C \in S_i\}$ is performed. Let $C' \in \mathcal{P}\setminus S_i$ be a cluster whose center $r_{C'}$ was reached by the exploration and let $r_C$ be the center in $R$ closest to $r_C'$. An edge $(r_C, r_{C'})$ with weight $d_G(r_C,r_{C'})$ is then added to the hopset. A supercluster $\hat{C}$ with center $r_{\hat{C}}=r_C$ is now created that contains all the vertices of $C$ and the clusters $C'$ for which a hopset edge was added.
  In the next stage of iteration $i$, all clusters within distance $\delta_i/2$ of each other that have not been superclustered at iteration $i$ will be interconnected. In other words, a \textit{separate} exploration of depth $\frac{\delta_i}{2}$ is performed from each such cluster center $r_C$ and if center of cluster $C'$ is reached, an edge $(r_C,r_C')$ with weight $d_G(r_C,r_{C'})$ will be also added to the hopset. The final phase of their algorithm only consists of the interconnection phase. We denote the hopset edges added for distance scale $(2^k,2^{k+1}]$ by $H_k$. For completeness, we review this algorithm in more detail and explain the exact parameters in Appendix \ref{app:hopset}. 

One important property of this hopset construction (proved in Lemma 3.3 of \cite{elkin2016}) that we will need for our analysis in Section \ref{sec:MPC}) is the following:
\begin{lemma}[\cite{elkin2016}] \label{lem:overlaps}
In the $i$-th iteration of a given distance scale $(2^k,2^{k+1}]$, for each node $v \in V$, w.h.p. the number of explorations of interconnection phase that visit $v$ is at most $O(deg_i \cdot\log n)$, where $deg_i$ is the sampling probability of the superclustering phase.
\end{lemma} 
  
Now we turn our attention to efficient construction of hopsets in distributed settings (such as CONGEST and Congested Clique) also proposed by \cite{elkin2016}. Note that each superclustering phase can be performed by a distributed Bellman-Ford exploration of depth $\delta_i$. For an interconnection phase, a separate distributed Bellman-Ford explorations of depth $\delta_i/2$ from cluster centers is performed. These Bellman-Ford algorithms can easily be implemented sequentially, however, in distributed settings, $O(n)$ rounds may be needed for each of the explorations of the larger scales. To overcome this issue, \cite{elkin2016} propose to use the hopsets $\cup_{\log \beta-1<j\leq k-1} H_j$, for constructing hopset edges $H_k$. More precisely, they observe that for any pair of nodes with distance less than $2^{k+1}$, hopsets $\cup_{\log \beta-1<j\leq k-1} H_j$ provide a $(1+\epsilon)$-stretch approximate shortest path with $2\beta+1$ hops between these pair of nodes. In other words, it is enough to run each Bellman-Ford exploration only for $O(\beta)$ rounds. 

\section{Distance Sketches in Massively Parallel Computation Model} \label{sec:MPC}
In this section we will focus on the MPC model. First we provide MPC algorithms for constructing distance sketches that have the same guarantees (with respect to the stretch/size tradeoff) as the centralized construction of Thorup-Zwick that run in polynomial (or slightly subpolynomial) time. Then in Section \ref{sec:extra_memory} we show how we can bring down the running time to polylogarithmic in exchange for a loss in accuracy.

First, we note that it is known from \cite{behnezhad2018} that for \textit{dense graphs} with $O(n^2)$ edges every Congested Clique algorithm (in which nodes use local memory of $O(n)$) can be implemented in the MPC$(n)$ model. Therefore, when memory per machine is $\Omega(n)$ and the graph is dense all the Congested Clique results discussed in Appendix \ref{app:congest_clique} also hold, except that we store the distance sketches rather than a central distance oracle. The more interesting case is when memory per machine is strictly sublinear in $n$. For the rest of this section we will turn our attention to the case where the memory is $n^\gamma$, where $0<\gamma \geq 1$ (i.e., strictly sublinear). For simplicity we assume that we can store the sketches in a single machine. Namely, we require $\tilde{O}(n^{1/k})$ memory per machine for stretch $O(k)$ distance sketches. This assumption can be relaxed (and in exchange the query algorithm will take $O(k)$ rounds instead of 2 rounds). 

 One main subroutine that we need is the \textit{restricted Bellman-Ford} algorithm. We then need to run many instances of this algorithm in parallel and handle other technicalities both for constructing hopsets, and then the distance sketches. First, we require following subroutines that will allow us to simulate one round of Bellman-Ford in MPC$(n^\gamma)$: 

\textbf{Sorting~\cite{goodrich2011}}. Given a set of $N$ comparable items, the goal is to have the items sorted on the output machines, i.e. the output machine with smaller ID holds smaller items.

\textbf{Indexing~\cite{andoni2018}}. Suppose we have sets $S_1, S_2 ,..., S_k$ of $N$ items stored in the system. The goal is to compute a mapping $f$ such that $\forall i \in [k], x \in S_i$, $x$ is the $f(S_i , x)$-th element of $S_i$. After running this algorithm the tuple $(x,f(S_i , x))$ is stored in the machine that stores $x$.

\textbf{Find Minimum ($x,y$).} Finds the minimum of $N$ values stored over a contiguous set of machines given ID $x$ of the first machine and ID $y$ of the last machine.  

\textbf{Broadcast ($b, x, y$).} Broadcasts a message $b$ to a contiguous group of machines given ID $x$ of the first machine and ID $y$ of the last machine.

The sorting and indexing subroutines can be performed in $O(1/\gamma)$ rounds of MPC$(n^\gamma)$ (\cite{andoni2018, goodrich2011}). We argue that we can solve the Find Minimum and Broadcast problems also in $O(1/\gamma)$ rounds of MPC$(N^\gamma)$ in the following theorem. At a high-level we use an \textit{implicit} aggregation tree of depth $O(\log_{N^\gamma}N)= \frac{1}{\gamma}$. 
\begin{theorem} \label{thm:find_min}
Given $N$ items over a contiguous range of machines $x$ to $y$, subroutines Find Minimum$(x,y)$ can be implemented in $O(1/\gamma)$ rounds of MPC$(N^\gamma)$. Moreover, the subroutine Broadcast$(x,y)$ can also be implemented in $O(1/\gamma)$ rounds of MPC$(N^\gamma)$.
\end{theorem}
\begin{proof}
We will first define a rooted \textit{aggregation tree} $\mathcal{T}$ with branching factor $N^\gamma$ where the machines $M_x,...,M_y$ are placed at the leaves (here $M_x$ denotes the machine with ID $x$). W.l.o.g assume that the machines in this range have increasing and sequential IDs. Note that we don't need to store this tree explicitly, and we only need each node to know its parent. Consider level $\ell$ of the tree (leaves have $\ell=0$). Each node in this level is a machine associated with the label $\ell$. For each node in level $\ell-1$ that has the $i$-th machine in its subtree, we set as its parent $M_{p(i, \ell)}$ where $p(i, \ell)= {x+\lfloor \frac{i}{N^{\ell\gamma}} \rfloor}$. Thus each machine can compute its parent given the label $\ell$. Similarly, each machine can compute the indices of its children (as a range). In other words, at each level $\ell$, we assign each group of $N^\gamma$ nodes of this tree to a parent node at level $\ell+1$.

The algorithm Find Minimum proceeds as follows: at each round $\ell$, each machine first computes minimum over its the values it knows, and then sends the outcome to the parent machine. Finally, the minimum will be computed and stored at the root machine, which may forward the value to another destination.
 The algorithm Broadcast will similarly use an aggregation tree, but this time it routes the message top-down. First message $b$ is sent to the first machine $M_x$, and then starting from $M_x$ in each round any machine that receives message $b$ sends this value to all of its children, which can be determined from the machine's ID and $y$. Eventually all the machines at the leaves will receive $b$. The number of rounds each of these subroutines take are the height of the aggregation tree which is $O(\log_{N^\gamma}N)= \frac{1}{\gamma}$.
\end{proof}

Running the (restricted) Bellman-Ford algorithm in MPC is not as straightforward as it is in the Congested Clique. One challenge is that for high-degree nodes, the edges corresponding to a single node are distributed over a set of machines. Therefore, for each round of Bellman-Ford these machines must communicate for computing and updating the distance estimates. Another hurdle is the fact that since nodes have different degrees, we do not have the range in which edges corresponding to a given node are stored a priori. To overcome these challenges we need to use the described subroutines, and for that we need to perform some preprocessing to append each edge with a tuple that we will describe shortly.

We will show how we can create and maintain the following setting: Given a graph $G=(V,E)$, the goal is to store all the edges incident to each node $v$ in a contiguous group of machines, which we denote by $M(v)$. More precisely, let $M_1,...,M_P$, where $P{=O(\frac{m}{n^\gamma})}$, be the list of machines ordered by their ID, and let $v_1,...,v_n$ be the list of vertices sorted by their ID. $M(v_i)$ consists of the $i$-th smallest contiguous group of machines, such that $|M(v_i)| = \lceil \frac{\deg(v_i)}{n^\gamma} \rceil$.  

Throughout the algorithm, let $M_{(u,v)}$ denote the machine that stores the edge $(u,v)$. Also, for all $u \in V$, let $r_u$ be the first machine in $M(u)$, and for any edge $(u,v) \in E$ let $i_u(v)$ be the index of $(u,v)$ (based on the lexicographic order) among all the edges incident to $v$. We need to compute and store the following information at $M_{(u,v)}$: $\deg(u), \deg(v)$, $r_u, r_v, i_u, i_v$ (here by storing $r_u$ we mean ID of $r_u$, and for simplicity we refer to $i_u(v)$ as $i_u$). We first explain how these labels can be computed for all edges in $O(\frac{1}{\gamma})$ rounds in the following lemma.
\begin{lemma}\label{lem:tuples}
 Let $M_{(u,v)}$ be the machine that stores a given edge $(u,v)$. We can create tuples of the form $((u,v), \deg(u), \deg(v)$, $r_u, r_v, i_u, i_v)$, stored at $M_{(u,v)}$ for all edges in $O(\frac{1}{\gamma})$ rounds in MPC$(n^\gamma)$, where $\gamma <1$.
\end{lemma} 
\begin{proof}
  Let $N(v)$ be the set of edges incident on node $v$. Without loss of generality, let us assume that both tuples of form $(u,v)$ and $(v,u)$ are present in the system for each edge and we assume $(u,v) \in N(u)$ and $(v,u) \in N(v)$ (note that the graph is still undirected). First, we use the indexing subroutine of \cite{andoni2018} on the sets $\{N(v)\}_{v \in V}$ to store index $i_u$ at $M_{(u,v)}$ and index $i_v$ at $M_{(v,u)}$. After this step tuples of form $((u,v), w(u,v), i_u)$ are stored at $M_{(u,v)}$.

 Then we sort the tuples based on edge IDs lexicographically, using sorting algorithm proposed in \cite{goodrich2011}. This will result in the setting described above in which edges incident to each node $u$ are stored in a contiguous group of machines $M(u)$. Now in order to compute $\deg(u)$, machines will check whether they are the last machine in $M(u)$ either by scanning their local memory or communicating with the next machine. Then the last machine in $M(u)$ sets $\deg(u)$ to the maximum index $i_u$ it holds. This machine can also compute $r_u$, ID of the first machine in $M(u)$ (using $\deg(u)$), and then broadcasts $\deg(u)$ and $r_u$ to all machines in $M(u)$. At the end of these computations, each tuple $((u,v), w(u,v), i_u)$ will be replaced by the tuple $((u,v), w(u,v), r_u, i_u, \deg(u))$. 
Next, we sort these tuples again but this time based on the ID of the smallest endpoint. In other words, for each edge $(u,v) \in E$, both tuples $((u,v), w(u,v), i_u, \deg(u))$ and $((v,u), w(v,u), i_v, \deg(v))$ will be at the same machine. Now we can easily merge these two tuples to create tuples of form $((u,v), w(u,v), i_u, i_v, \deg(u), \deg(v))$. 
\end{proof}

After computing the tuples, we use the sorting subroutine again to redistribute the edges into the initial setting of having contiguous group of machines $M(u)$ for all $u \in V$. After these preprocessing steps, we are ready to perform updates required for the restricted Bellman-Ford algorithm. A summary of this algorithm is presented in Algorithm \ref{alg:Bellman-Ford_MPC}.

\begin{algorithm}
\caption{Restricted Bellman-Ford in MPC$(n^\gamma)$.}
\label{alg:Bellman-Ford_MPC}
\SetKwInOut{Input}{Input}
\SetKwInOut{Output}{Output}
\Input{ Graph $G=(V,E)$ distributed among machines $M_1,...,M_P$ and source $s$.}
\Output{$h$-hop restricted distances from the source $s$ to all nodes $u \in V$, $d^h(s,v)$.}
Create the tuple $((u,v), i_u, i_v, r_u, r_v, \deg(u), \deg(v))$ at $M_{(u,v)}$ for each edge $(u,v) \in E$ (by Lemma \ref{lem:tuples}).\\
Sort the edges lexicographically so that edges incident to $v$ are stored in a contiguous group of machines $M(v)$ (by \cite{goodrich2011}).\\
\For{$i=0$ to $h$}{
  \For{$v \in V$}{
	   Compute $\hat{d}(s,v)$ by finding (using Theorem \ref{thm:find_min} $\min_{u \in N(v)} \hat{d}(s,u)+ w(u,v)$).\\
	   Broadcast updated distances to everyone in $M(v)$ (also by Theorem \ref{thm:find_min}).\\
	   Each machine in $M_{(v,u)}$ sends $\hat{d}(s,v)$ to $M_{(u,v)}$ (located at $r_u + \lfloor \frac{i_u}{n^\gamma} \rfloor$).
   }
}
\end{algorithm}


\begin{theorem} \label{thm:bellman-ford_MPC}
Given a graph $G=(V,E)$ and a source node ${s \in V}$ the restricted Bellman-Ford algorithm (Algorithm \ref{alg:Bellman-Ford_MPC}) computes distances $d^h(s,v)$ for all $v \in V$ in $O(\frac{h}{\gamma})$ rounds of MPC$(n^\gamma)$.
\end{theorem}
\begin{proof}
 After storing the tuples $(i_u, i_v, r_u, r_v, \deg(u), \deg(v))$ at $M_{(u,v)}$ for each $(u,v) \in E$, the restricted Bellman-Ford algorithm proceeds as follows: in each round, for each node $v$, we first find the minimum distance estimate for $v$ and send it to $r_v$. Then $r_v$ will broadcast the minimum distance found to all the machines in $M(v)$. By Theorem \ref{thm:find_min} both of these operations take $O(1/\gamma)$ rounds. Then for each $(v,u) \in N(v)$, $M_{(v,u)}$ sends the updated distance directly to $M_{(u,v)}$, which islocated at index $r_u + \lfloor \frac{i_u}{n^\gamma} \rfloor$. All the operations for each of the $h$ iterations of Bellman-Ford take $O(1/\gamma)$ rounds. 
\end{proof}

 We now need to argue that hopsets of \cite{elkin2016} can be constructed in MPC$(n^\gamma)$. We show this in the following theorem. Here we assume that the weights are polynomial in $n$, which is not unrealistic since in MPC the total memory is assumed to be $\tilde{O}(m)$ bits.

\begin{theorem}\label{thm:hopsets_MPC}
For any graph $G = (V, E, w)$ with $n$ vertices, and parameters $\rho \leq \gamma \leq 1,1 \leq \kappa \leq (\log n)/4, 1/2 > \rho \geq 1/\kappa$ and $0 < \epsilon < 1$, there is an algorithm in MPC$(n^\gamma)$ model that computes a $(\beta, \epsilon)$-hopset with expected size $O(n^{1+\frac{1}{\kappa}} \log n)$ in $O(\frac{n^\rho}{\rho} \cdot \log^2 n \cdot \beta)$ rounds whp, where $\beta= O((\frac{\log n}{\epsilon}  \cdot (\log \kappa + 1/\rho))^{\log \kappa+ \frac{1}{\rho}})$.
\end{theorem}
\begin{proof}
As explained in Appendix \ref{app:hopset}, the distributed implementation of this algorithm just performs multiple restricted Bellman-Ford algorithms in each phase. Recall also that it is enough to run each of the Bellman-Ford instances only for $O(\beta)$ rounds, by using the fact that for constructing hopset edges $H_k$ for a distance scale of $(2^k,2^{k+1}]$, the hopsets $\cup_{\log \beta-1<j\leq k-1} H_j$ can be used recursively.

Each round of a single Bellman-Ford algorithm can be simulated in $O(\frac{1}{\gamma})$ rounds of MPC$(n^\gamma)$ by running the algorithm of Theorem \ref{thm:bellman-ford_MPC} on each node, whose edges may be distributed over multiple machines. Hence each superclustering phase can be performed in $O(\frac{\beta}{\gamma})$ rounds. But at each interconnection phase multiple separate Bellman-Fords will run from each cluster center remaining. Thus we need to argue that these runs of Bellman-Ford will not violate the memory (and IO memory) limit of each machine. This can be shown using Lemma \ref{lem:overlaps}, which states that for each vetex $v \in V$, w.h.p. the number of explorations of interconnection phase that visit $v$ is at most $O(\deg_i \cdot\log n)$. In other words, each node only forwards messages to at most $O(\deg_i \cdot\log n)$ in each depth $\delta_i/2$ Bellman-Ford explorations performed for an interconnection phase. Moreover, the parameters of their construction is set so that $\deg_i=O(n^\rho)$ throughout the algorithm (see Appendix \ref{app:hopset} for more details). Hence, each node $v \in V$ need to store and forward distance estimates corresponding to at most $O(n^\rho \log n)$ sources for $O(\log(\kappa \rho)+ \frac{1}{\rho})$ iterations, and each Bellman-Ford runs for $O(\beta)$ rounds. These separate Bellman-Ford runs can be pipelined. Overall, all of the Bellman-Ford explorations can be implemented in ${O(\frac{\beta}{\gamma} \cdot n^\rho \log n)}$.
\end{proof}

 We can now construct a hopset first and then run the distributed variant of the algorithm in Section \ref{sec:TZ} due to \cite{sarma2015} for constructing the distance sketches on the new graph. The sketch of a given node $v$ can be stored at a machine in $M(v)$. 

\begin{proof}[Proof of Theorem \ref{intro_thm:MPC_sketches}]
After constructing a $(\beta, \epsilon)$-hopset (by setting $\kappa =k$), we store the edges added to each node $v$ by redistributing them among machines $M(v)$ that simulate $v$. Let $G'=(V, E \cup H, w')$ be the graph obtained by adding hopset edges. For constructing distance sketches with stretch $2k-1$, we run the algorithm of \cite{sarma2015} (described in Appendix \ref{app:congest_clique}) on $G'$. We run the restricted Bellman-Ford algorithm (Algorithm \ref{alg:Bellman-Ford_MPC}) in $O(\frac{\beta}{\gamma})$ rounds. Overall, $O( \frac{\beta n^\rho \log^2 n}{\rho \gamma})$ rounds are needed for the hopset construction (by Theorem \ref{thm:hopsets_MPC}), and $O(kn^{1/k}\log n \cdot \frac{\beta}{\gamma})$ rounds for building the distance sketches on $G'$. 
In case $k=O(1)$ we set $\rho=1/\kappa$, and $\kappa=k$ to get $\beta=\tilde{O}(\log(n))$ and total running time $\tilde{O}(n^{1/k})$. In case $k= \Theta(\log n)$, we will set $1/\kappa=\rho= \sqrt{\frac{\log \log n}{\log n}}$. Note that in this case we need $\tilde{O}(n^{1+\rho})$ space for constructing the hopsets, but after this step the size of the distance oracle stored will be $\tilde{O}(n)$.
\end{proof}

\subsection{Polylogarithmic Round Complexity} \label{sec:extra_memory}
In this section we describe how we can modify our algorithm to run in a polylogarithmic number of rounds in exchange for increasing the stretch.   We do this by first constructing a spanner, which sparsifies the graph (``shrinking" the input) and thus allows us to act as if we have ``extra" total space. It turns out that this extra space is incredibly powerful, and will let us build distance sketches in polylogarithmic time.  But in the end we have to pay for both the stretch of the spanner and the stretch of the sketch, so we only achieve stretch $O(k^2)$ rather than stretch $2k-1$ for sketches of size $\tilde O(n^{1/k})$.

There are intuitively two reasons why this extra space is so helpful.  First, in MPC having extra space (or extra machines) is equivalent to having larger total communication bandwidth.  This intuitively allows us to speed up the main construction algorithm by running the Bellman-Ford algorithms ``in parallel".  There are some technical details 
but it is not surprising that extra bandwidth is helpful.  

The second reason why extra space is helpful is less obvious.  Goodrich et al.~\cite{goodrich2011} gave a powerful simulation argument, showing that PRAM algorithms can be efficiently simulated in MPC as long as the total number of processors used and the total space used by the PRAM algorithm are bounded by the size of the input.  This is a very useful theorem, but the requirement that the number of processors is only the size of the input is very restrictive.  For example, the state of the art PRAM algorithms for constructing hopsets use $\Omega(mn^{\rho})$ processors rather than $O(m)$ (for some value $\rho$ determined by the parameters of the hopset).  It turns out to be easy to extend~\cite{goodrich2011} to show that if we have extra total space, we can use that extra space and communication to simulate PRAM algorithms that use slightly more processors or space.  Thus by using a spanner first to sparsify the input, we give ourselves extra space and thus the ability to efficiently simulate a wider class of PRAM algorithms (hopsets in particular).
\paragraph{MPC with Extra Space.}
 First we define a variant of MPC with extra machines (and thus extra space) denoted by MPC$(S, S')$ where $S$ is memory per machine, the number of machines is $\Theta(\frac{mS'}{S})$  and $m$ is the total input size. This also implies the total memory available is $\Theta(mS')$ rather than $\Theta(m)$. We are first going to analyze our algorithm in this variant of MPC, and then switch back to the standard setting.

In \cite{goodrich2011} it was shown that with a small overhead PRAM algorithms can be simulated in MPC under certain assumptions on the number of processors and the memory used. We use a simple extension of their result for our new MPC variant.
\begin{theorem}\label{thm:PRAM_extra} 
Given a PRAM algorithm using $\mathcal{P}=O(m \alpha)$ processors that runs in time $\mathcal{T}$, and uses $O(m \alpha)$ total memory at any time, this algorithm can be simulated in $O(\mathcal{T}/\gamma)$ rounds of MPC($m^\gamma, \alpha$), for any $0<\gamma<1$.
\end{theorem}

This stronger variant of MPC also lets us extend Theorem \ref{thm:find_min} for larger message sizes. We define a generalized variant of Find Minimum that takes a collection of vectors and computes their coordinate-wise minimum, and a generalizes version of Broadcast which broadcasts a vector of messages (rather than just a single message). We get the following lemma. 
\begin{lemma}\label{lem:mod_broadcast}
We can compute generalized Find Minimum$(x,y)$ over $N$ vectors of length $\alpha$ stored on a contiguous range of machines $x$ to $y$ in $O(1/\gamma)$ rounds of MPC$(N^\gamma, \alpha)$. Moreover, the generalized Broadcast$(\mathbf{b}, x, y)$ subroutine can also be implemented in $O(1/\gamma)$ rounds. 
\end{lemma}
\begin{proof}
 In the new settings we have $\Theta(N^{1-\gamma} \cdot \alpha)$ machines that can be used for computation over $N$ items in range $(x,y)$, rather than $\Theta(N^{1-\gamma})$ machines used in Theorem \ref{thm:find_min}. Therefore we can assign each coordinate to a group of $N^{1-\gamma}$ machines and then use a similar aggregation tree argument as in Theorem \ref{thm:find_min} on all the coordinates in parallel in $O(1/\gamma)$ rounds for both problems. 
\end{proof}


Next, we describe how the algorithm of Theorem \ref{thm:distance_sketch} can be modified to utilize the extra resources in MPC$(n,n^{1/k}\log n)$ to improve the round complexity. We use an argument similar to \cite{sarma2015} with a few changes. The complete argument can be found in Appendix \ref{app:extra_memory_sketches}.

\begin{theorem}\label{thm:extra_distance_sketches}
Given a graph $G=(V,E)$ with shortest path diameter $\Lambda$, there is an algorithm in MPC$(n^\gamma, n^{1/k}\log n)$ that runs in time $O(k\Lambda)$ w.h.p.~and constructs Thorup-Zwick distance sketches of size $O(kn^{1/k}\log n)$ with stretch $2k-1$.
\end{theorem}

A straightforward extension of Theorem \ref{thm:extra_distance_sketches} implies that given a $(\beta, \epsilon)$-hopset for a graph, we can compute distance sketches with stretch $(1+\epsilon)(2k-1)$ in $O(\frac{\beta}{\gamma})$ rounds of MPC$(n^\gamma, n^{1/k}\log n)$. 
 Next, we show that in addition to proving Theorem \ref{thm:extra_distance_sketches}, the extra memory also lets us improve the number of rounds for the hopset construction. To show this, we use a result in \cite{elkin2016} that constructs hopsets in PRAM, which is as follows:
\begin{theorem}[\cite{elkin2016}] \label{thm:hopset-PRAM}
For any graph $G = (V, E, w)$ with $n$ vertices, and parameters $2 \leq \kappa \leq (\log n)/4, 1/2 > \rho \geq 1/\kappa$ and $0 < \epsilon < 1$, there is a PRAM algorithm that computes a $(\beta, \epsilon)$-hopset with expected size $O(n^{1+\frac{1}{\kappa}} \log n)$ in ${O(\frac{1}{\rho} \cdot \log^2 n \cdot \log \kappa \cdot \beta)}$ PRAM time whp, where $\beta= O(\frac{\log n (\log \kappa + 1/\rho)}{\epsilon})^{\log \kappa+ \frac{1}{\rho}}$ using $\tilde{O}( (m+n^{1+1/\kappa})n^\rho)$ processors. 
\end{theorem}
We now argue that by having more space/machines, we are can implement the algorithm in Theorem \ref{thm:hopset-PRAM} with the same guarantees in low-memory MPC settings.
 We will not discuss the details of the PRAM construction but the intuition here is similar to Theorem \ref{thm:extra_distance_sketches}. At a high level, having more communication/memory will allows us to perform all the $\tilde{O}(n^\rho)$ Bellman-Ford explorations required in the algorithm of Theorem \ref{thm:hopset-PRAM} in parallel. 
\begin{corollary}\label{cor:hopsets_extra}
For any graph $G = (V, E, w)$, and parameters $0 < \epsilon < 1, 1/\kappa < \gamma \leq 1, \kappa \geq 2$, there is an algorithm that computes a $(\beta, \epsilon)$-hopset with size $O(n^{1+\frac{1}{\kappa}} \log n)$ w.h.p.~in ${O( (\kappa/\gamma) \cdot \log^2 n \cdot \log \kappa \cdot \beta)}$ rounds of MPC($n^\gamma, n^{1/\kappa}$), where $\beta= O(\frac{\log n (\log \kappa)}{\epsilon})^{\log \kappa+\kappa+1}$. 
\end{corollary}
\begin{proof}
The claim directly follows by setting $\rho= 1/\kappa$ in Theorem \ref{thm:hopset-PRAM} and then applying the simulation in Theorem \ref{thm:PRAM_extra} in MPC$(n^\gamma, n^{1/\kappa})$ on the new graph.
\end{proof}

\paragraph{Obtaining Extra Space.} Our modified algorithm for MPC$(n^\gamma)$ now proceeds as follows: we first construct a spanner, then construct a hopset on this spanner, and then use Theorem \ref{thm:extra_distance_sketches}. Intuitively, by sparsifying the graph we can ``buy" more memory and hence more communication. In other words, by building a spanner we can extend the results of the extra memory setting to the standard MPC setting.

There are several efficient PRAM algorithms for constructing spanners that we can simulate in MPC. We use an algorithm proposed by \cite{baswana2007} that constructs a $(2k-1)$-spanner of size $O(kn^{1+1/k}\log n)$ with high probability. We then use Theorem \ref{thm:PRAM_extra} with $\alpha=1$ (i.e. the original simulation of \cite{goodrich2011}) to construct the spanner in $O(\frac{k}{\gamma}\log n \log^*n)$ rounds of MPC($n^\gamma$), and then redistribute the spanner edges (e.g., by sorting), to make the input distribution uniform over all the machines. We can now put everything together to get the polylogarithmic construction.

\begin{proof}[ Proof of Theorem \ref{thm:polylog_distancesketches}]
 We first construct a $4k-1$-spanner with size $O(kn^{1+\frac{1}{2k}})$. We denote this spanner by $G'$. Since $G'$ has size $m'=O(n^{1+\frac{1}{2k}})$, while our total memory (and consequently overall communication bound) is still based on the original graph. Equivalently, the number of machines is $\frac{m}{n^\gamma}= \Omega(\frac{m' n^{1/{2k}}\log n}{n^\gamma})$ (since $m=\Omega(kn^{1+1/k} \log n)$), and therefore we are exactly in the MPC$(n^\gamma, n^{\frac{1}{2k}})$ setting, but where the input graph is $G'$. Then we use Corollary \ref{cor:hopsets_extra} to construct a $(\beta, \epsilon)$-hopset for $G'$ with $\beta={O(\frac{k}{\gamma} \cdot (\frac{ \log n \cdot \log k}{\epsilon} )^{\log k +1 +k})}$ rounds of MPC$(n^\gamma)$. Finally, after adding the hopset edges to $G'$ we use Theorem \ref{thm:extra_distance_sketches}. The new stretch is clearly ${O(k^2(1+\epsilon))}$.
\end{proof}
\subsection{Single-source shortest path} \label{sec:SSSP}
In various models (such as PRAM, CONGEST and Congested Clique) hopsets are used for solving shortest path problems (e.g. \cite{cohen2000, henzinger2016, elkin2016}), and thus it is natural to see how they can be used for this application in the MPC model. In particular, we discuss application of Theorem \ref{thm:hopsets_MPC} in solving the (approximate) single-source shortest path problem. As stated earlier, while this problem is well-studied in many distributed models, including the Congested Clique model, we are not aware of any non-trivial results for this problem in the low memory MPC setting.
\begin{theorem}
Given a weighted undirected graph $G=(V, E, w)$, a source node $s \in V$, and $0 < \gamma \leq 1, 0 < \epsilon<1$ we can compute $(1+\epsilon)$-approximate distances from $s$ to all nodes in $V$ w.h.p.~in $O(\frac{1}{\gamma}) \cdot 2^{\tilde{O}(\sqrt{\log n})}$ rounds of MPC with $\Theta(n^\gamma)$ memory per machine.
\end{theorem}
\begin{proof}
We first construct a hopset using Theorem \ref{thm:hopsets_MPC} by setting $\rho= \sqrt{\frac{\log n}{\log \log n}}$, and $\kappa = \Theta(\log n)$. This will let us build a hopset with hopbound $2^{\tilde{O}(\log n)}$ in time $O(\frac{1}{\gamma}) \cdot 2^{\tilde{O}(\log n)}$. We then run the restricted Bellman-Ford algorithm (Algorithm \ref{alg:Bellman-Ford_MPC}) in $O(\frac{1}{\gamma}) \cdot 2^{\tilde{O}(\sqrt{\log n})}$ rounds of MPC($n^\gamma$). The idea behind this choice of parameters is the following: any attempt to improve the running time by getting a smaller hopbound (e.g. constant) will increase the time required to construct the hopset. In other words, this choice of parameters will make the time required for preprocessing (construction of the hopset) almost the same as the time required for running the Bellman-Ford algorithm.
\end{proof}

Finally, we show that we can used the technique in Section \ref{sec:extra_memory} to find constant approximation to single source shortest path in polylogarithmic time for graphs with a certain density. In particular, by first constructing a spanner and then using Corollary \ref{cor:hopsets_extra}, we can also solve $4k (1+ \epsilon)$-approximate SSSP (for any $2 \leq k \leq O(\log n)$) on any graph with $m= \Omega(n^{1+1/k} \log n)$ edges in fewer number of rounds. After constructing a $4k-1$-spanner, we construct a $(\beta, \epsilon)$-hopset for an appropriate hopbound $\beta$ using the extra space and then run a single restricted Bellman-Ford (Algorithm \ref{alg:Bellman-Ford_MPC}) from the source in $O(\beta/\gamma)$ rounds of MPC$(n^\gamma)$. By setting $\kappa=k$ we get,
\begin{corollary}\label{cor:constant_SSSP}
For any graph $G = (V, E, w)$ with $n$ vertices, $m= \Omega(n^{1+1/k})$ edges, and $0 < \epsilon < 1, 1/k < \gamma \leq 1, k >2$, and a source node $s \in V$, there is an algorithm that w.h.p. finds a $4k(1+\epsilon)$-approximation of shortest path distance from $s$ to all nodes ~in ${O(\frac{1}{\gamma} \cdot (\frac{ \log n \cdot \log k}{\epsilon} )^{\log k +k+1 })} $ rounds of MPC($n^\gamma$). In particular, for $k=O(1)$ the algorithm runs in ${O(\frac{1}{\gamma} \cdot (\frac{ \log n}{\epsilon} )^{O(1)})}$ rounds.
\end{corollary}

\bibliography{DistributedDistanceOracles}
\newpage
\appendix
\section{Query algorithm} \label{app:query}
In this appendix section, we briefly review the (sequential) query algorithm of \cite{thorup2005}. Given sketches of a pair of nodes $(u,v) \in V$ the query algorithm proceeds as follows: For each $0 \leq i \leq k-1$, we check if $p_i(u) \in B_i(v)$ or $p_i(v) \in B_i(u)$. Let $j$ be the smallest level at which one of these conditions occur. Note that by construction $p_{k-1}(u) \in B(v)$ and $p_{k-1}(v) \in B(u)$, and this implies that $j \leq k-1$ exists. Then if the first condition holds, the distance estimate $\tilde{d}(u,v)= d(u, p_j(u))+d(v,p_j(u))$ and if the second conditions holds we set $\tilde{d}(u,v)= d(u, p_j(v))+d(v,p_j(v))$. Note that these distance are stored with the sketch and can be computed. This clearly takes $O(k)$ time (sequentially), and it can be shown (see \cite{thorup2005}) that this estimate satisfies $\tilde{d}(u,v)=(2k-1)d(u,v)$-stretch.
\section{Distance Oracles in Congested Clique} \label{app:congest_clique}
In this section, we will explain how the distance oracle can be constructed in the Congested Clique model. We will use the algorithm described in Section \ref{sec:TZ} by Das Sarma et al.~\cite{sarma2015} that constructs Thorup-Zwick distance \textit{sketches} with stretch $2k-1$ and size $kn^{1+1/k}$ in $O(k\Lambda n^{1/k})$ rounds in the CONGEST model, where $\Lambda$ is shortest-path diameter. Our algorithm is similar to their algorithm, with the difference that we first construct a hopset. This will allow us to terminate the algorithm earlier while preserving the distances within a $(1+\epsilon)$ factor. Constructing hopsets in the Congested Clique model can be done more efficiently than in CONGEST model. Hence, unlike the known algorithms in the CONGEST model, we can build a distance oracle in time independent of the shortest path diameter.

First we formally state a theorem proved in \cite{elkin2016} for hopset construction in Congested Clique. 
\begin{theorem}[\cite{elkin2016}] \label{thm:hopsets}
For any graph $G = (V, E, w)$ with $n$ vertices, and parameters $2 \leq \kappa \leq (\log n)/4, 1/2 > \rho \geq 1/\kappa$ and $0 < \epsilon < 1$, there is a distributed algorithm for the Congested Clique model that computes a $(\beta, \epsilon)$-hopset with expected size $O(n^{1+\frac{1}{\kappa}} \log n)$ in ${O(\frac{n^\rho}{\rho} \cdot \log^3 n \cdot \beta)}$ rounds whp, where $\beta= O(\frac{\log(n)\cdot (\log \kappa + 1/\rho)}{\epsilon \cdot \rho})^{\log \kappa+ \frac{1}{\rho}}$. 
\end{theorem}

 Roughly speaking, adding a $(\beta, \epsilon)$-hopset edges is as if the shortest path diameter reduced to $\beta$ in exchange for a small loss in the stretch. In other words, hopsets will let us cut of distance computation after exploring $\beta$ hops. Later on we will explain how we can set the parameters $\rho$ and $\kappa$ depending on the stretch parameter for the distance oracle, to get our desired running time.

We need the \textit{h-restricted} distributed Bellman-Ford subroutine (Algorithm \ref{alg:Bellman-Ford}) which is widely used in previous work on distributed distance estimation (e.g.~see \cite{lenzen2013}, or \cite{sarma2015}).  We use the following lemma that follows from basic properties of Bellman-Ford algorithm.
\begin{lemma}
There is a distributed variant of the Bellman-Ford algorithm runs in $O(h)$ rounds in Congested Clique and for all nodes $u \in V$, computes $d_h(s,u)$, the length of the shortest path between $s$ and $u$ among the paths that have at most $h$ edges.
\end{lemma}
In order to compute the shortest path from $s$ to all nodes, we will need to set $h= \Lambda$, the shortest path diameter. But this can be as large as $\Omega(n)$. Hence we will use a $(\beta, \epsilon)$-hopset to approximately find the distance in $O(\beta)$ time only. In other words, by constructing a $(\beta, \epsilon)$-hopset $H$, we would know that there is a path of hopbound $\beta$ with length $(1+\epsilon) d(u,v)$ among any pair of nodes $u,v \in V$, and hence Algorithm \ref{alg:Bellman-Ford} can approximate the distances $d(s,v)$ up to a factor of $(1+\epsilon)$ for all $v \in V$.

We now argue that a distance oracle can be constructed by first preprocessing the input graph by constructing a $(\beta, \epsilon)$-hopset (by Theorem \ref{thm:hopsets}) and then running the algorithm of \cite{sarma2015} that was described in Section \ref{sec:TZ} for $O(\beta)$ rounds. Let us first state the result of \cite{sarma2015} in the following theorem.


\begin{algorithm}
\caption{Distributed Bellman-Ford with hopbound $h$.}
\label{alg:Bellman-Ford}
\SetKwInOut{Input}{Input}
\SetKwInOut{Output}{Output}
\Input{Undirected weighted graph $G=(V, E, w)$, and source node $s \in V$.}
\Output{$h$-hop restricted distances from the source $s$ to all nodes $u \in V$, $d^h(s,v)$.}
	Set $\forall v \in V: \hat{d}(s,v)= \infty$.\\  
	\For{Rounds $i=0$ to $h$}{
		\For{$\forall v \in V$}{
		\If{$\exists u \in N(v): \hat{d}(s,v)> \hat{d}(s,u)+ w(u,v)$ }{
	 		Set $\hat{d}(s,v):=\min_{u \in N(v)} (\hat{d}(s,u)+ w(u,v))$, and send $\hat{d}(s,v)$ to all neighbors.}
	 		}	
	}
\end{algorithm}
Next we review the distance oracle algorithm for the Congested Clique model in Algorithm \ref{alg:distanceoracle}. This algorithm was proposed by \cite{sarma2015} for constructing distance sketches in the CONGEST model.

\begin{algorithm}
\caption{Preprocessing distributed distance oracle for stretch $2k-1$ due to \cite{sarma2015}.}
\label{alg:distanceoracle}
\SetKwInOut{Input}{Input}
\SetKwInOut{Output}{Output}
\Input{Undirected graph $G=(V, E, w)$, and a coordinator node.}
\Output{Approximate distance oracle stored at the coordinator.}
Set $A_0=V, A_k=\emptyset$.\\
\For{every $v \in V$}
 {\For{$i=1$ to $k-1$}
	{If $v \in A_{i-1}$ with probability $n^{-1/k}$ add $v$ to $A_i$.}
}
\For{$i=k-1$ down to $1$}{
		Coordinator runs Algorithm \ref{alg:Bellman-Ford} out of set $A_i$.\\
	\For{$\forall v \in V$}{
	 Set $p_i(v)= \text{argmin}_{u \in A_i} d(u,v)$, and $d(v,A_i)= d(p_i(v),v)$.}}
	\For{$w \in A_{i}\setminus A_{i+1}$}{
		Coordinator runs the algorithm in Theorem \ref{thm:distance_sketch}.}
\end{algorithm}

In \cite{sarma2015} the following result was shown for Algorithm \ref{alg:distanceoracle}:
\begin{theorem}[\cite{sarma2015}]\label{thm:distance_sketch} 
Given undirected graph $G=(V, E, w)$ with shortest path diameter $\Lambda$, there is an algorithm that runs in $\tilde{O}(\Lambda \cdot kn^{1/k} \log n)$ rounds of Congested Clique w.h.p.~and outputs a Thorup-Zwick distance oracle with stretch $(2k-1)$ at the coordinator with high probability.
\end{theorem}
Note that the algorithm in \cite{sarma2015} is for the CONGEST model, which we can easily implement in the Congested Clique model. In other words, we are not using the extra power of the Congested Clique model here, rather, we will use this power for constructing hopsets more efficiently. Moreover, in \cite{sarma2015} distance \textit{sketches} are constructed at each node. It is easy to see that nodes can then send their sketches to the coordinator within a constant factor of total number of rounds required to build a distance oracle consisted of the sketches for all nodes. 
 Next, we will utilize the hopset construction of \cite{elkin2016} to make the preprocessing algorithm more efficient with respect to time and message complexity. Let $G'=(V, E \cup H, w')$ be the graph obtained by adding a $(\beta, \epsilon)$-hopset $H$ to the undirected graph $G=(V, E, w)$. By running the algorithm in Theorem \ref{thm:distance_sketch} on $G'$ we will get the following result.
\begin{corollary} \label{cor:distance_oracle_time}
Given a graph $G=(V,E, w)$ and a $(\beta, \epsilon)$-hopset $H$ for $G$, there is an algorithm that runs in $\tilde{O}( \beta \cdot kn^{1/k})$ rounds of Congested Clique and outputs a Thorup-Zwick distance oracle with stretch $(2k-1) (1+\epsilon)$ on the graph $G'=(V, E \cup H, w')$ at the coordinator with high probability.
\end{corollary}
 Next we will analyze the message complexity of algorithm of Theorem \ref{thm:distance_sketch} and show that w.h.p. $\tilde{O}(kmn^{1/k} \beta)$ messages need to be exchanged. It is not hard to see that the number of messages exchanged for constructing a $(\beta, \epsilon)$-hopset is $\tilde{O}(\beta n^{1+\rho}/\rho)$ (this follows by analysis of \cite{elkin2016}). Hence the dominant number of messages exchanged is for running the algorithm of Theorem \ref{thm:distance_sketch}.
 
\begin{lemma}\label{lem:communication}
Total number of messages exchanged for constructing a Thorup-Zwick distance oracle (with parameters specified in Corollary \ref{cor:distance_oracle_time}) on graph ${G'=(V, E \cup H, w')}$ is w.h.p.~$O( \beta m \cdot kn^{1/k}\log n)$.
\end{lemma}
\begin{proof}
The algorithm of Corollary \ref{cor:distance_oracle_time} runs in $O(\beta  k n^{1/k}\log n)$ rounds w.h.p.~and overall for each edge in the graph $O(1)$ messages are exchanged.
\end{proof}
 We now combine the hopset construction and Theorem \ref{thm:distance_sketch} together to obtain our main result. We will use Theorem \ref{thm:hopsets} to construct a hopset $H$ on graph $G=(V, E, w)$, and then run Algorithm \ref{alg:distanceoracle} on the obtained graph $G'=(V, E \cup H, w')$ and get the following:
\begin{theorem}\label{thm:main}
Given a graph $G=(V, E, w)$, polynomial weights\footnote{Same as in other models this assumption can be relaxed using techniques of \cite{elkin2016}  in exchange for extra polylogarithmic factors.} and parameters $2 \leq \kappa \leq (\frac{\log n}{4})$, ${1/\kappa \leq \rho \leq 1/2}, 0 < \epsilon<1$, we can construct a Thorup-Zwick distance oracle with stretch $(2k-1)(1+\epsilon)$ and size $O(kn^{1+\frac{1}{k}}\log n)$ w.h.p.~in $O(\beta (\frac{n^{\rho}}{\rho}\cdot \log^3n  +  n^{\frac 1k} \log n))$ time, where ${\beta= O(\frac{\log(n)\cdot (\log \kappa + 1/\rho)}{\epsilon})^{\log \kappa+ \frac{1}{\rho}}}$.
\end{theorem}

The running time depends both on the parameter $\rho$ and stretch $k$. In other words, there is a tradeoff between the stretch $k$ and the running time of this algorithm. When stretch $k$ is smaller, we can choose a larger value for $\rho$ and the dominant part of the running time would still be the distance oracle construction. On the other hand, for larger values of stretch $k$, since the distance oracle construction algorithm can be performed more efficiently, we need to set $\rho$ smaller to balance out the running time of constructing a hopset and that of constructing the distance oracle over the new graph. The parameter $\kappa$ mostly just impacts the hopset size and the constant factor in the exponent of hopbound $\beta$. Let us consider two special cases of $k=O(1)$ and $k=\Omega(\log n)$ to understand these bounds better. In the special case of $k= \Omega(\log(n))$ the hopset construction step takes more time, and so we use the recent result of \cite{censor2019} for the hopset construction to get a polylogarithmic running time. They construct a hopset of size $\tilde{O}(n^{3/2})$ with hopbound $O(\log^2(n)/\epsilon)$ in $O(\log^2(n)/\epsilon)$ rounds. 

\begin{corollary}\label{cor:congested_clique_stretch}
Given a graph $G=(V,E, w)$, and constant $0 < \epsilon \leq 1$, we can construct a Thorup-Zwick distance oracle with stretch $(2k-1)(1+\epsilon)$ in the Congested Clique model, s.t.,
\begin{itemize}
\item In case $k=O(1)$, w.h.p.~we require $\tilde{O}(n^{1/k})$ rounds.
\item In case $k=\Omega(\log n)$, w.h.p.~we require $\tilde{O}(\log(n))$ rounds.
\end{itemize}
\end{corollary}
\begin{proof}
For stretch $k=O(1)$ we use Theorem \ref{thm:main} and set $\rho=1/\kappa$, and $\kappa=k$ to get $\beta=\tilde{O}(\log(n))$ and total running time $\tilde{O}(n^{1/k})$. 
In case $k= \Theta(\log n)$, we will set $1/\kappa=\rho= \sqrt{\frac{\log \log n}{\log n}}$. In both cases, by setting $\rho$ to be a smaller constant, such as $\rho=1/2$ we can have a smaller $\beta$ (but still polylogarithmic), but the preprocessing algorithm will use the larger space of $\tilde{O}(m+n^{1+\rho})$ space and communication, rather than $\tilde{O}(m+n^{1+1/k})$. 
In the special case $k=\Omega(\log(n))$, we use the hopset algorithm of \cite{censor2019}, which takes polylogarithmic time.

\end{proof}

\paragraph{Communication Reduction with Spanners.}
In this section, we will describe how spanners can be used as a tool for reducing communication in exchange for an extra factor in the stretch.  Recall, A $t$-spanner of a graph $G$ is a subgraph $H$ such that $d_G(u,v) \leq d_H(u,v) \leq d_G(u,v)$ for all $u,v \in V$.  We will use the spanner construction of \cite{baswana2007} which computes spanners efficiently in the more restricted CONGEST model. 
\begin{theorem}[\cite{baswana2007}]\label{thm:spanners}
For any weighted graph, a $(2t-1)$-spanner of expected size $O(tn^{1+1/t})$ can be computed in the CONGEST model in $O(t^2)$ rounds and $O(tm)$ message complexity.
\end{theorem}
This construction allows us to turn the input graph for algorithms described in this section into a sparser graph. By doing so we will lose a factor of $t$ in the approximation ratio but we only need to run algorithm of Theorem \ref{thm:distance_sketch} on a graph with $O(n^{1+1/t})$ edges. Hence, we first run the Algorithm of Theorem \ref{thm:spanners} to get a spanner $G_t$, and then run the algorithm of Theorem \ref{thm:main} on $G_t$. Then by Lemma \ref{lem:communication} we have,
\begin{theorem} \label{thm:CC-communication}
Given a graph $G=(V, E, w), t,k >1$, we can construct a Thorup-Zwick distance oracle of size $O(kn^{1+1/k}\log n)$ with stretch $t\cdot (2k-1)(1+\epsilon)=O(kt)$ w.h.p.~with total communication of $\tilde{O}(kn^{1/t+1/k} \beta +tm)$, where $\beta$ and the running time are the same as in {Theorem \ref{thm:main}}.
\end{theorem}
This implies that there is a direct tradeoff between the approximation ratio and the amount of communication when size of the distance oracle is fixed. In other words, when $n^{1/t} =o(m)$ the amount of communication required for distance oracles of stretch $O(kt)$ is smaller than the amount required for building distance oracles of stretch $O(k)$, where the size is in both cases $O(kn^{1+1/k}\log n)$. 
\section{Hopset construction of \cite{elkin2016}} \label{app:hopset}
In this section we provide more details on the hopset construction of \cite{elkin2016}. In their (sequential) algorithm, they consider each distance scale $(2^k, 2^{k+1}], k=0,1,2,...$ separately. In other words, for a given $k \leq \log \Lambda$ they construct a set of edges $H_k$ such that for each pair $u,v \in V$ where $2^k \leq d_G(u,v) \leq 2^{k+1}$, there is a path of length $(1+ \epsilon) d_G(u,v)$ with hop-bound $\beta$ in $E \cup H_k$. For $k < \log \beta$, it is enough to have $H_k= \emptyset$.

 Next, we consider a fixed distance scale $R \in (2^k,2^{k+1}], k \geq \log \beta$. Each iteration of the algorithm consists of a set of \textit{superclustering}, and \textit{interconnection} phases, except that there will be no superclustering in the last phase. Initially, the set of clusters is $\mathcal{P}=\{ \{v\}_{v \in V}\}$. Each cluster in $C \in \mathcal{P}$ has a cluster center which we denote by $r_C$. The algorithm uses a sequence $\delta_1,\delta_2,...$ of distance thresholds and a sequence $\deg_1,\deg_2,...$ of degree thresholds that determines the sampling probability of clusters (we describe details of these sequences below). At the $i$-th iteration, every cluster $C \in \mathcal{P}$ is sampled with probability $1/\deg_i$. Let $S_i$ denote the set of sampled clusters.
 Now a single shortest-path exploration of depth $\delta_i$ from the set of centers of sampled clusters $R=\{r_C \mid C \in S_i\}$ is performed. Let $C' \in \mathcal{P}\setminus S_i$ be a cluster whose center $r_{C'}$ was reached by the exploration and let $r_C$ be the center in $R$ closest to $r_C'$. An edge $(r_C, r_{C'})$ with weight $d_G(r_C,r_{C'})$ is then added to the hopset. A supercluster $\hat{C}$ with center $r_{\hat{C}}=r_C$ is now created that contains all the vertices of $C$ and the clusters $C'$ for which a hopset edge was added.
  In the next stage of iteration $i$, all clusters within distance $\delta_i/2$ of each other that have not been superclustered at iteration $i$ will be interconnected. In other words, a \textit{separate} exploration of depth $\frac{\delta_i}{2}$ is performed from each such cluster center $r_C$ and if center of cluster $C'$ is reached, an edge $(r_C,r_C')$ with weight $d_G(r_C,r_{C'})$ will be also added to the hopset. The final phase of their algorithm only consists of the interconnection phase.

The algorithm depends on two parameters $\kappa \geq 2$, which controls the hopset size, and $\frac{1}{\kappa} \leq \rho\leq \frac{1}{2}$, which impacts the running time and the hopbound. This algorithm has two stages as follows: the first stage has $\log(\kappa \rho)$ phases in which the degree sequence grows exponentially, while the second stage has $O(1/\rho)$ phases in which the degree sequence remains the same. More precisely, if phase $i$ is in the first stage, we have $\deg_i=n^{2^i/\kappa}$, whereas $\deg_i=n^\rho$ if $i$ is a phase in the second stage. The growth rate of distance thresholds remains the same in both stages (it increases by a factor of $1/\epsilon$). In particular we have $\delta_i = \epsilon^{\ell-i} \cdot 2^{k+1} + 4R_i$, where $\ell$ is total number of phases and $R_{i+1}= \delta_i+ R_i$. 

In \cite{elkin2016} it is shown that this construction results in a $(\beta, \epsilon)$-hopset of size $O(n^{1+1/\kappa} \log n)$, where $\beta = O(\frac{\log n+ \frac{1}{\rho}}{\epsilon})^{\log \kappa +O(\frac{1}{\rho})}$. They also show how a similar algorithm can be implemented in multiple distributed setting by using hopset edges for smaller distance scales in construction of larger distance scales. 
Specifically, each superclustering phase can be performed by a distributed Bellman-Ford exploration of depth $\delta_i$. However for an interconnection phase, a separate distributed Bellman-Ford explorations of depth $\delta_i/2$ from cluster centers is performed, each of which could take $\Omega(n)$ rounds in distributed settings. To overcome this issue, \cite{elkin2016} propose to use the hopsets $\cup_{\log \beta-1<j\leq k-1} H_j$, for constructing hopset edges $H_k$. More precisely, they observe that for any pair of nodes with distance less than $2^{k+1}$, hopsets $\cup_{\log \beta-1<j\leq k-1} H_j$ provide a $(1+\epsilon)$-stretch approximate shortest path with $2\beta+1$ hops between these pair of nodes. In other words, it is enough to run each Bellman-Ford exploration only for $O(\beta)$ rounds. 

One main property that this construction has is shown in Lemma 3.3 of \cite{elkin2016} that states for each vertex $v \in V$, w.h.p. the number of explorations of interconnection phase that visit $v$ is at most $O(\deg_i \cdot\log n) = O(n^\rho \cdot \log n)$. We also use this property to show that this hopset construction cab be implemented in the MPC model efficiently. More formally,
\begin{lemma}[\cite{elkin2016}] \label{lem:overlaps}
In the hopset algorithm of \cite{elkin2016}, during the $i$-th iteration of a given distance scale $(2^k,2^{k+1}]$, for each node $v \in V$, w.h.p. the number of explorations of interconnection phase that visit $v$ is at most $O(deg_i \cdot\log n)$, where $deg_i$ is the sampling probability of the superclustering phase.
\end{lemma} 
  


\section{Proofs Omitted from Section \ref{sec:MPC}}\label{app:MPC}
\subsection{Proof of Theorem \ref{thm:extra_distance_sketches}}  \label{app:extra_memory_sketches}
Next, we show how in the MPC setting with extra memory we can improve the number of rounds, using an argument similar to \cite{sarma2015}.
\begin{theorem}
Given a graph $G=(V,E)$ with shortest path diameter $\Lambda$, there is an algorithm in MPC$(n^\gamma, n^{1/k}\log n)$ that runs in time $O(k\Lambda)$ w.h.p.~and constructs Thorup-Zwick distance sketches of size $O(kn^{1/k}\log n)$ with stretch $2k-1$.
\end{theorem}
\begin{proof}
The algorithm is as follows: we have $k$ \textit{phases} for each level of Thorup-Zwick. Sampling sets $A_{k-1} \subseteq ...\subseteq A_1$ is straightforward. We start from the $k$-th phase, and we run Bellman-Ford (Algorithm \ref{alg:Bellman-Ford_MPC}) with the following modification: each node $u$ keeps a vector of size $O(n^{1/k} \log n)$ of distance estimates $\tilde{d}(v,u)$ for all $v \in B(u)$. Then we run modified variants of the Broadcast and Find Min subroutines (Lemma \ref{lem:mod_broadcast}) to update distances $\tilde{d}(v,u)$ based on a message received from a neighbor $u' \in N(u)$ if and only if $\tilde{d}(v,u)+w(u,u') < d(u,A_{i+1})$ and $\tilde{d}(v,u')+w(u,u') < \tilde{d}(v,u)$. 

Based on Lemma 5 in \cite{sarma2015}, we know that at the end of phase $i$, each node $u \in V$ knows $B_i(u)$ and its distance to all nodes in $B_i(u)$. In particular, inductively each node $u$ knows $d(u,A_{i+1})$ before starting phase $i$. Note that algorithm of \cite{sarma2015} keeps a queue for all possible source nodes for their scheduling. We do not have space to store such a queue for all nodes. Here we simply only store a map of size $O(kn^{1/k} \log n)$ for each node $u$ that corresponds to distance estimates for all $v \in B(u)$.

Next, we argue that each phase takes $O(\frac{\Lambda}{\gamma})$ rounds based on an inductive argument similar to Lemma 6 in \cite{sarma2015}. Let $v \in B_i(u)$, and assume that there is a shortest path with $j$ hops between $v$ and $u$ which we denote by $v=v_0,...,v_j=u$. We use an induction on $j$. In the base case, $u$ and $v$ are neighbors and in $O(1/\gamma)$ rounds the aggregations can be performed as in Theorem \ref{thm:find_min}. By inductive hypothesis $v_{j-1}$ received a message after $O(\frac{(j-1)}{\gamma})$ rounds. If $v_{j-1}$ had found its shortest path to $v$ before the $(j-1)$-st iteration of Bellman-Ford, it would have already sent an update to $v_j$. Otherwise, $v$ computes and broadcasts the updated distance using $O(1/\gamma)$ rounds of MPC$(n^\gamma, n^{1/k}\log n)$. We showed in Lemma \ref{lem:mod_broadcast} that this can be done in parallel for all messages corresponding to sources in $B(u)$ in $O(1/\gamma)$ rounds. Hence $v_j$ receives the updated distance after $O(j/\gamma)$ rounds, where $j \leq \Lambda$ by definition. Finally, all nodes will receive the distances from nodes in their bunches after $O(k\Lambda/\gamma)$ rounds.
\end{proof}

\section{Distance Oracles in the Streaming Model}\label{app:streaming}
In this section we will describe how the Thorup-Zwick distance oracles can be constructed in the insert-only streaming model. For graph problems, the stream is a sequence of edges (and their weights), and the goal is to solve the problem in space strictly sublinear in number of edges. For some problems we might need to see multiple \textit{passes} of the stream. Similar to the distributed settings, we will use the hopset construction of \cite{elkin2016}. They show that in streaming settings a $(\beta, \epsilon)$-hopset, where with the following guarantees can be constructed. 
\begin{theorem}[\cite{elkin2016}] \label{thm:hopsets_streaming}
For any graph $G = (V, E, w)$ with $n$ vertices, and any $2 \leq \kappa \leq (\log n)/4, 1/2 > \rho \geq 1/\kappa, 1 \leq t \leq \log n$ and $0 < \epsilon < 1/2$, there is a streaming algorithm that computes a $(\beta, \epsilon)$-hopset with expected size $O(n^{1+\frac{1}{\kappa}} \log^2 n)$, where $\beta= O(\frac{1}{\epsilon}  \cdot (\log(\kappa) + \frac{1}{\rho}) \log n))^{\log(\kappa)+ \frac{1}{\rho}}$ requiring either of the following resources:
\begin{itemize}
\item $O(\beta \log n)$ passes w.h.p.~and expected space $O(\frac{n^{1+\rho}}{\rho} + n^{1+\frac{1}{\kappa}} \log^2 n)$,
\item $O(n^\rho \cdot \beta \cdot \log^2 n)$ passes w.h.p.~and expected space $O(n^{1+\frac{1}{\kappa}}  \log^2 n)$.
\end{itemize} 
\end{theorem}

We will next explain how the distance oracle can be constructed in $O(\beta)$ passes given a hopset with hopbound $\beta$.
First, we need a variant of \textit{restricted} Bellman-Ford for streaming settings. The idea of using Bellman-Ford in streaming settings has been previously used for shortest path computation (e.g. \cite{elkin2016}, \cite{henzinger2016}). This algorithm is similar to the distributed variant: on receipt of each edge $(u,v) \in E$ we will check to see if the distance from any of the sources in $S$ should be updated. After $i$ passes of the algorithm,  all nodes have the $i$-restricted distance to nodes in $s$.
The restricted Bellman-Ford in streaming is presented in Algorithm \ref{alg:bellman-ford-stream}. Note that unlike centralized Bellman-Ford nodes do not store and initial distance estimate (due to space limitation in the streaming model). This algorithm uses $O(|S|\cdot nh)$ total space. 

\begin{algorithm}[h]
\caption{Restricted Bellman-Ford in the Streaming Model}
\label{alg:bellman-ford-stream}
\SetKwInOut{Input}{Input}
\SetKwInOut{Output}{Output}
\Input{Undirected weighted graph $G=(V, E, w)$, and source node $s \in V$.}
\Output{$h$-hop restricted distances from the source $s$ to all nodes $u \in V$, $d^h(s,v)$.}
\For{$O(h)$ passes}{
	\For{$(u,v) \in E$}{
					\If{$\hat{d}(s,v) = \emptyset$ or $\hat{d}(s,u) + w(v,u)	< \hat{d}(s,v)$}{		
						$\hat{d}(s,v)= \hat{d}(s,u) + w(v,u)$}
					\If{$\hat{d}(s,v) = \emptyset$ or $\hat{d}(s,v) + w(v,u)	< \hat{d}(s,u)$}{		
						$\hat{d}(s,u)= \hat{d}(s,v) + w(v,u)$}					
	}	
	}
\end{algorithm}

Using the restricted Bellman-Ford algorithm, we can construct a Thorup-Zwick distance oracle of stretch $(2k-1)(1+\epsilon)$ in $O(\beta)$ passes. The details of this algorithm is presented in Algorithm \ref{alg:distanceoracle_stream}.

\begin{algorithm}[h]
\caption{Preprocessing distance oracle of stretch $2k-1$ in the streaming model.}
\label{alg:distanceoracle_stream}
\SetKwInOut{Input}{Input}
\SetKwInOut{Output}{Output}
\Input{ Undirected graph $G=(V, E, w)$ of shortest path diameter $\Lambda$.}
\Output{ Approximate distance oracle.}
Set $A_0=V, A_k=\emptyset$.\\

\For{$i=1$ to $k-1$}
	{If $v \in A_{i-1}$ with probability $n^{-1/k}$ add $v$ to $A_i$.}
	
		Run Algorithm \ref{alg:bellman-ford-stream} \textit{in parallel} out of each set $A_i, 1 \leq i \leq k,$ to find $p_i(v)= \text{argmin}_{u \in A_i} d(u,v)$, and set $d(v,A_i):= d(p_i(v),v)$.

\For{$O(\Lambda)$ passes}{
	\For{$(u,v) \in E$}{
		\For{$i=k-1$ down to $1$}{	
			\For{$s \in A_{i}\setminus A_{i+1}$}{
				\If{ $\hat{d}(s,v) < d(v,A_{i+1})$ or $\hat{d}(s,u) +w(u,v) < d(v,A_{i+1})$ }{
					\If{$\hat{d}(s,u) + w(v,u)	< \hat{d}(s,v)$}{		
						$\hat{d}(v,s)= \hat{d}(s,u) + w(v,u)$}
						}
				\If{ $\hat{d}(s,u) < d(u,A_{i+1})$ or $\hat{d}(s,v) +w(u,v) < d(u,A_{i+1})$ }{
					\If{$\hat{d}(s,u) + w(v,u)	< \hat{d}(s,u)$}{		
						$\hat{d}(s,u)= \hat{d}(s,v) + w(v,u)$}					
					}
			}
		}
	}
	}
\end{algorithm}


Next, we will explain how distance oracles can be constructed in $O(\Lambda)$ passes, where $\Lambda$ is the shortest path diameter. We will then use a hopset of hopbound $\beta$ to reduce the number of passes to $O(\beta)$. This algorithm is again similar to the distributed algorithm. Sets $A_1,..,A_{k-1}$ can easily be sampled in sublinear space. Here again for finding the distances from each set $A_i$ to all nodes, we will add a virtual node $a_i$ and add an edge of weight $0$ between $a_i$ and all the nodes in $A_i$. We then run the restricted Bellman-Ford algorithm from each of these sources $a_i$ separately. This phase requires $O(kn \Lambda)$ space and $O(\Lambda)$ passes. In the final phase, we need to find the distance from each node in $s \in A_i \setminus A_{i+1}$ to all nodes in $C(s)= \{ v \mid v \in B(s)\}$. We will run a variant of the Bellman-Ford algorithm in which each node $v$ only stores a distance only if $\hat{d}(s,v) < d(v,A_{i+1})$ or if this condition holds after receiving an update from a neighbor. We will get the following lemma. 

\begin{lemma} \label{lem:streaming_space}
There is an algorithm that runs in $O(\Lambda)$ passes, and w.h.p.~constructs a $2k-1$ stretch Thorup-Zwick distance oracle of size $O(kn^{1+1/k} \log n)$ using $O(kn^{1+1/k} \log n)$ total space.
\end{lemma}
\begin{proof}
It is clear that described algorithm takes $O(\Lambda)$ passes, and correctly updates all the distances required for building a Thorup-Zwick distance oracle. We also show that the space required is the same as the distance oracle size. This follows from the fact that for each node $v$, we are only storing distances to the nodes that are in $v$'s bunch $B(v)$, and we know from \cite{thorup2005} that w.h.p.~$|B(v)|= O(kn^{1/k} \log n)$. Thus the total space is w.h.p.~$O(kn^{1+1/k} \log n)$. 
Similar to the distributed case, given a graph $G=(V,E,w)$, we can use the $(\beta, \epsilon)$-hopset construction of Theorem \ref{thm:hopsets_streaming} to obtain a graph $G'=(V, E \cup H, w')$ which has shortest path diameter $O(\beta)$, and the distances in $G$ are preserved up to a factor of $(1+\epsilon)$. Then by running Algorithm \ref{alg:distanceoracle_stream} on $G'$, we would require $O(\beta)$ passes to build a distance oracle with stretch $(2k-1)(1+\epsilon)$.
\end{proof}

 We set the parameters in such a way that we have space to store all the hopset edges locally. Thus while running the algorithm of Theorem \ref{thm:distance_oracle_streaming} we also consider the hopset edges to decide when to update the distance. However, for readability of our algorithm, here we assume that the hopset edges are also appearing in the stream. Hence, by first running the hopset construction algorithm of Theorem \ref{thm:hopsets_streaming}, and then running the algorithm of Theorem \ref{thm:distance_oracle_streaming}, we will get the following result:
\begin{theorem} \label{thm:distance_oracle_streaming}
Given a graph $G=(V, E,w)$, there exists a streaming algorithm that constructs a Thorup-Zwick distance oracle of stretch $(2k-1)(1+ \epsilon)$ of size $O(kn^{1+1/k} \log n)$ w.h.p.~using either of the following resources\footnote{All the bounds expressed in expectation can be turned into high probability bound with an additional factor of $\log n$ in the number of passes.}:
\begin{itemize}
\item $O(\beta \log n)$ passes w.h.p.~and expected space $O(\frac{n^{1+\rho}}{\rho} + n^{1+\frac{1}{k}}\log^2 n)$,
\item $O(n^\rho \cdot \beta \cdot \log^2 n)$ passes w.h.p.~and expected space $O(n^{1+\frac{1}{k}}\log^2 n)$, 
\end{itemize}
 where $\beta= O(\frac{1}{\epsilon}  \cdot (\log(k) + 1/\rho) \log n))^{\log(k)+ \frac{1}{\rho}}$ and ${\frac{1}{k} \leq \rho \leq \frac{1}{2}}$.
\end{theorem}

In particular, when $k=O(1)$ we will use the first case of Theorem \ref{thm:distance_oracle_streaming} and set $\rho=1/k$, and when $k=\Omega(\log n)$ we will use the second case and set $\rho= \sqrt{\frac{\log \log n}{\log n}}$.  We have,
\begin{corollary}
Given a graph $G=(V, E,w)$, there exists a streaming algorithm that constructs a Thorup-Zwick distance oracle of stretch $(2k-1)(1+ \epsilon)$ of size $O(kn^{1+1/k} \log n)$  w.h.p. such that:
\begin{itemize}
\item If $k=O(1)$, we require $O(\log^k n)$ passes and expected space $O(n^{1+1/k} \cdot \log^2 n)$ with high probability.
\item If $k=\Omega(\log n)$, we require $2^{\tilde{O}(\sqrt{\log n})}=n^{o(1)}$ passes,  with high probability.
\end{itemize}
\end{corollary}

\section{Alternative methods} \label{app:alternatives}
In the Congested Clique, rather than computing a Thorup-Zwick distance oracle we could instead compute a graph spanner and store this at the coordinator node. A $(2k-1)$-spanner of $G$ is simply a subgraph which preserves distances up to a $(2k-1)$ factor, so once such spanner is at the coordinator, a classical centralized shortest-path algorithm would yield a distance estimate that is accurate up to $(2k-1)$ (as with our distance oracle). A similar approach can be used in the streaming model. While a reasonable approach, there are a few drawbacks.  

First, the local computation time becomes superlinear, rather than $O(k)$ as in our oracle.  While computation is generally extremely cheap compared to network communication, there is still an enormous gap between superlinear and $O(k)$ (since $k$ is at most logarithmic).  And for large graphs, this may indeed rise to the level of network delay time scales.

The more important drawback, though, is that spanners cannot be used in the MPC model.  Even an extraordinarily sparse spanner would not fit into the memory of a single server in low-memory MPC, so the spanner would (just like the original graph) have to be stored in a distributed fashion.  So we would still have the same problem that we started with: how to compute distance estimates in a distributed graph.  Only distance sketches allow us to answer such queries in such a small number of rounds (in particular, two rounds after the sketches have been computed).  

Another direction that one could take is running an all-pairs shortest path algorithm (APSP). This approach has multiple drawbacks: First, for fast queries we will need to store the whole adjacency matrix, which clearly uses much more space. Secondly, such algorithms are slower and use more resources. 

\end{document}